%% file: 00-main.tex
\pdfoutput=1
\RequirePackage{amsmath}
\documentclass[runningheads]{llncs}
\usepackage[T1]{fontenc}
\usepackage{adjustbox}
\usepackage{latexsym}
\usepackage{amssymb}
\usepackage{amsmath}

\usepackage{float}
\usepackage{multirow}
\usepackage{booktabs}
\usepackage{makecell}
\usepackage{color}
\usepackage{xcolor}
\usepackage{tabularx}
\usepackage{tasks}
\usepackage[inline]{enumitem}
\usepackage{tikz}
\usepackage{multicol}
\usepackage{circuitikz}
\usepackage{setspace}

\newcommand{\revision}[1]{\textcolor{black}{#1}}
\usepackage{graphicx}
\usepackage{epstopdf}
\usepackage{subfig}
\usepackage{makecell}
\usepackage{multirow}
\usepackage{booktabs}
\usepackage{array}

\graphicspath{ {figures/} }
\DeclareGraphicsExtensions{.eps,.pdf,.jpeg,.png}

\usepackage[ruled,vlined,linesnumbered]{algorithm2e}
\input{macros}

\makeatletter
\renewcommand\paragraph{\@startsection{paragraph}{4}{\z@}%
	{-12\p@ \@plus -4\p@ \@minus -4\p@}%
	{-0.5em \@plus -0.22em \@minus -0.1em}%
	{\normalfont\normalsize\bfseries}}
\makeatother

\usepackage{array}
\newcolumntype{C}{>{\centering\arraybackslash}X}

\usepackage{notation}

\usepackage{ifdraft}
\usepackage{comment}

\ifoptionfinal{%
	\excludecomment{todocomment}
}{%
	\usepackage[obeyFinal]{todonotes}
}

\begin{document}

\title{Multi-agent Path Finding for Timed Tasks
	\\ using Evolutionary Games}

% %
% %\titlerunning{Abbreviated paper title}
% % If the paper title is too long for the running head, you can set
% % an abbreviated paper title here
% %

\author{Sheryl Paul\inst{1} \and
	Anand Balakrishnan\inst{1} \and
	Xin Qin\inst{1} \and
	Jyotirmoy V. Deshmukh \inst{1}}
\authorrunning{Paul et al.}
% First names are abbreviated in the running head.
% If there are more than two authors, 'et al.' is used.
%
\institute{University of Southern California, Los Angeles California, USA- 90007
	\email{\email{\{sherylpa,anandbal,xinqin,jdeshmuk\}@usc.edu}}\\
}
% \author{Anonymous Authors}

\maketitle              % typeset the header of the contribution
\begin{abstract}
	% \begin{todo}
	%    {Abstract needs to be re-written entirely, after intro is written}
	% \end{todo}

	Autonomous multi-agent systems such as hospital robots and package delivery
	drones often operate in highly uncertain environments and are expected to
	achieve complex temporal task objectives while ensuring safety. While
	learning-based methods such as reinforcement learning are popular methods to
	train single and multi-agent autonomous systems under user-specified and
	state-based reward functions, applying these methods to satisfy trajectory-level
	task objectives is a challenging problem. Our first contribution
	is the use of weighted automata to specify trajectory-level objectives, such
	that, maximal paths induced in the weighted automaton correspond to desired
	trajectory-level behaviors. We show how weighted automata-based specifications
	go beyond timeliness properties focused on deadlines to performance properties
	such as expeditiousness. Our second contribution is the use of evolutionary game
	theory (EGT) principles to train homogeneous multi-agent teams targeting
	homogeneous task objectives. We show how shared experiences of agents and
	EGT-based policy updates allow us to outperform state-of-the-art reinforcement learning (RL) methods in minimizing path length by nearly 30\% in large spaces. We also show that our algorithm is computationally faster than deep RL methods by at least an order of magnitude. Additionally our results indicate that it scales better with an increase in the number of agents as compared to other methods.
\end{abstract}
\section{Introduction}
\input{intro.tex}

\section{Preliminaries}

\input{preliminaries}
\section{Reward Shaping for timed Multi-Agent Reach-Avoid Problems}
\input{solution_approach}

\section{Policy Optimization}
\input{egt_framework}

\section{Experiments and Results}
\input{experiments}
\section{Related Work}
\input{related}
\section{Conclusion}
\input{conclusion}
\newpage
\bibliographystyle{splncs04}
\bibliography{references}
\end{document}

%% file: macros.tex
\usepackage{amsmath,amssymb}

\newcommand{\initialstates}{\mathcal{I}}

\usepackage{wrapfig}

\DeclareMathOperator*{\EX}{\mathbb{E}}

\newcommand{\agidx}{i}

\newcommand*{\agentstate}{s}

\newcommand{\action}{a}

\newcommand{\timeid}{t}

\newcommand{\utility}{{\utility}^{\agidx}_{\timeid}}

\newcommand{\policy}{\pi}

\newcommand{\goalset}{F}
\newcommand{\obstacles}{O}

\newcommand{\traj}{\tau}

%% file: intro.tex
% \begin{verbatim}
% - Motivation
%   o Many future systems will be multi-agent
%   o An important subclass is homogeneous multi-agent systems
%   o Task objectives for such systems can be specified using
%     temporal logics
    
% - Planning for Single-agent environments with temporal logic objectives
%   is well-researched (RL + STL papers)

% - Limited work on planning for Multi-agent systems (Yash Vardhan Pant,
%   Navid's arXiv submission)

% - Main issue in RL + STL world for single-agent: convergence to optimal
%   policies, guaranteeing satisfaction of objectives

% - Autonomous agents, Decentralized execution, centralized training, 
%   a couple of lines about MARL

% - Summary of contributions:

% \end{verbatim}

% The realm of MAPF is broad, encompassing tasks in
% surveillance and patrolling, \cite{uav,wildfire,agmon_multiagent,agmon_patrol}
% warehouse operations \cite{sven_warehouse,config}, coordination of industrial
% tug robots \cite{sven_amazon}, and logistical planning for railways and aviation
% \cite{sven_flatland,sven_airport}. 

Large-scale deployment of multi-agent autonomous mobile systems is becoming a
reality in many sectors such as automated warehouses
\cite{sven_warehouse,config}, surveillance and patrolling
\cite{agmon_multiagent,agmon_patrol}, package delivery
\cite{salzman2020research}, and logistics support for cargo in aviation, and
railways \cite{sven_airport,abate}. On the one hand, these autonomous systems
are safety-critical and require careful planning to avoid unsafe events such
as collisions with other agents or the environment, transgressing zoning
restrictions, etc. On the other hand, the deployed systems are expected to be
high-performing; for instance, where performance is measured in the number of
tasks performed, task completion times, and load balancing across agents. This
is an especially difficult challenge when the agents do not know the location of
the goal (making it difficult to use heuristic search methods). The core
technical challenge for such systems is the {\em planning} problem: how do we
synthesize plans for the agents to move in their environment while satisfying
safety specifications and task objectives with unknown goal locations, while
simultaneously ensuring performance or throughput of the overall system?

\begin{wrapfigure}[23]{r}{0.45\textwidth}
    \centering
    \includegraphics[width=\linewidth]{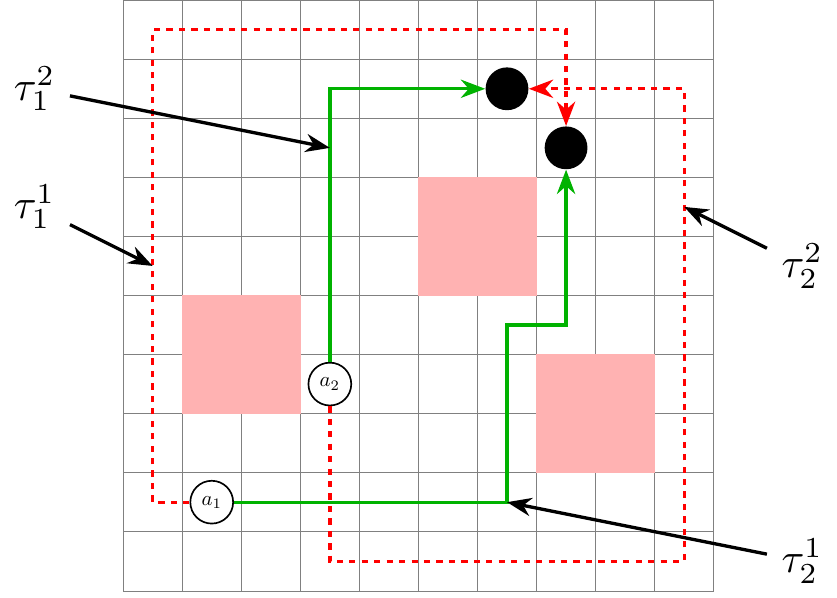}
    \caption{\textit{Shortest paths in the grid environment: agents \(a_1\) and \(a_2\) must get to designated final states (black circles) while avoiding red obstacle regions. Red dashed trajectories \(\tau^1_1\) of agent \(a_1\) and \(\tau^2_1\) of \(a_2\)  satisfy the task but are not the shortest paths, green
    solid trajectories  \(\tau^1_2\) and \(\tau^2_2\)are the shortest paths. Arrows
    point towards goals.}}
    \label{fig:intro_grid_sample}
    \vspace{10pt}
\end{wrapfigure}

Traditional AI approaches such as the {\em multi-agent pathfinding} (MAPF)
problem use various heuristic methods in combinatorial optimization and graph
search to plan paths for many agents to reach their respective goals while
avoiding collisions \cite{mapf_defn}. However, these approaches involve a
centralized planner with full observability of the system, as well as full
knowledge of the goal locations. These techniques usually suffer from poor
scalability with increasing numbers of agents and world sizes. For highly
uncertain environments, an alternative is to endow agents with autonomy that
allows them to navigate an environment with the help of sensors. In this
setting, each individual agent has a control policy to determine the optimal
action to perform based on its sensor observations. Reinforcement learning (RL)
has emerged as a promising approach to learning such control policies in uncertain
environments modeled as Markov decision processes. The main idea in RL is to use
hand-crafted and state-based rewards, and learn policies that optimize
the cumulative long-term rewards. An important assumption is that additive
long-term rewards capture desired global behavior. There are many instances of
{\em reward hacking}, where myopic rewards for contrasting objectives lead to
undesired long-term behaviors \cite{amodei2016concrete}.

\revision{An approach to tackle this problem of specifying time-sensitive and sequential tasks with complex dependencies is to use logic-based specifications like in Linear Temporal Logic (LTL) \cite{pnueli1977temporal} or its extensions, including Signal Temporal Logic (STL) \cite{maler2005real} and Mission-time Linear Temporal Logic (MLTL) \cite{reinbacher2014temporallogic}.
Moreover, such specification languages can be used to define complex tasks for multi-agent systems \cite{kempa2022swarm}.}
% Recent work has thus focused on the use of logic-based specifications in Linear
% Temporal Logic (LTL) or Signal Temporal Logic (STL) to articulate precise task
% descriptions for time-seeditive and sequential tasks with complex dependencies.
% Though formal specifications capture desired long-term behavior, they pose a
% challenge to RL algorithms that are designed to use state-based rewards. 
In recent years, {\em reward shaping} using temporal logic-based objectives
\cite{rltl1,rltl2,rltl3,xinloki} has emerged as a promising approach to address
the problem of defining state-based rewards to satisfy complex spatiotemporal tasks.
% this problem. 
A prevailing theme in many reward shaping papers is to utilize
either the Boolean satisfaction semantics for LTL \cite{abate} or quantitative
satisfaction semantics (i.e., robustness values for STL)
\cite{balakrishnan2019structured,calin} as reward signals. 

However, while these
logical specifications are well-suited for specifying timeliness properties such
as deadlines, they are not as effective at specifying performance properties
such as {\em expeditiousness}. For example, consider the environment shown in
Fig.~\ref{fig:intro_grid_sample} and the trajectories shown in green and red,
respectively. Let $a_i$ denote the position of the \(i^{\text{th}}\) agent, then
the formula $\Ev_{[0,25]} (a_i \in \mathrm{Goal})$ specifies that the agent
should reach the goal within $25$ time-steps.

While all the trajectories satisfy the formula, it is clear that the green
trajectories are more expeditious as they reach the goal faster. Note that we
do not want to replace this specification with one that indicates a shorter
deadline. Rather, we would like to prefer the ones that reach the goal faster.
Such requirements are crucial in multi-agent systems where there may be hard
deadlines on agents reaching their goals, with agents reaching goals at varying
times, and we want to prefer behaviors where average expeditiousness across
agents is preferred.

To address this problem, we use \emph{weighted automata} specifications as they allow us to capture performance metrics like expeditiousness \cite{droste2009weighted,mohri2009weighted}.
Weighted automata are a class of automata with transition \emph{weights} interpreted over algebraic structures and generalize the various qualitative and quantitative semantics of discrete-time temporal logics \cite{algebraicstl,jaksic2018quantitative}.

 In this paper, we present an innovative application of an {\em evolutionary algorithm} rooted in evolutionary game theory (EGT) \cite{egt_games} to address the Multi-Agent Pathfinding (MAPF) problem. This algorithm is adept at deriving optimal policies for multiple agents in complex environments, advancing beyond existing optimization-focused techniques that rely on preset goal knowledge and environmental models \cite{pant1,pant2,hashemi2024scaling,sam}.
% \begin{todo}
%     {Put something in about MLTL here and also in the Related Work}
% \end{todo}
We examine a scenario involving numerous homogeneous agents without predefined goal locations, where a single, unified policy learned from collective experiences supersedes individual training. This shifts the problem to a single-agent learning paradigm. Our method is benchmarked against single-agent algorithms like Q-learning \cite{sutton}, PPO \cite{ppo}, and heuristic searches such as $A^*$ \cite{astar}. Here, $A^*$ benefits from its heuristic approach to goal navigation. Our results show that our evolutionary strategy not only achieves optimal policy faster but also competes with, and can surpass the aforementioned methods in efficiently reaching targets.
\paragraph{Contributions:}
To summarize, the main contributions of this paper are as follows:
\begin{enumerate*}[label={(\arabic*)}]

\item we introduce weighted automata (WA)-based task objectives and show that
these encourage agents to finish their tasks sooner;
    
\item we propose an algorithm `MAPF-EGT' (Multi Agent Pathfinding using
Evolutionary Game Theory) that uses evolutionary techniques to learn stochastic
control policies for homogeneous agents trying to fulfill homogeneous task
objectives (but with {\em a priori} unknown goals); and
    
\item we show empirical results that our method outperforms state-of-the-art RL
and heuristic search methods in various grid-world environments. 
    
\end{enumerate*}

%% file: preliminaries.tex
\paragraph{Policy synthesis for stochastic games:}
To formalize the problem of synthesizing policies for agents such that they
satisfy their assigned task specifications, we use the framework of stochastic
games.

\begin{definition}[Stochastic Game] A {\em stochastic game} over a finite set of
agents $\mathcal{N}$, where each agent $i \in \mathcal{N}$ can be modeled as a
tuple $(\mathcal{S}^i, \mathcal{A}^i, \Delta^i, \gamma, \mathcal{R}^i)$, where
$\mathcal{S}^i$ denotes the set of possible states of agent $i$, $\mathcal{A}^i$
is the set of actions available to agent $i$, and $\Delta^i: \mathcal{S}^i
\times \mathcal{A}^i \times \mathcal{S}^i \to [0,1]$ is a joint probability
distribution over states, actions and next states that defines the transition
dynamics of the game. At each time-step $t$, agent $i$ is in some state $s_t^i$
and executes action $a_t^i$, transitioning to the next state $s_{t+1}^i \sim
\Delta^i(s'\mid s = s_t^i, a = a_t^i)$, and receives reward $r^i_t =
\mathcal{R}^i(s_t^i,a_t^i)$. The discount factor $\gamma \in [0,1]$ prioritizes
early rewards by discounting a reward $r^i_t$ at time $t$, by a factor of
$\gamma^t$. A {\em joint action} at time $t$ is a tuple $\mathbf{a}_t = (a_t^1,
a_t^2, \ldots, a_t^n)$. The system transitions from a joint state $\mathbf{s}_t
= (s_t^1, s_t^2, \ldots, s_t^n)$ based on the joint action $\mathbf{a}_t$ to a
joint state sampled from the joint transition probability distribution
$\mathbf{s}_t^i \sim \Delta(\mathbf{s}' | \mathbf{s}, \mathbf{a})$.
\end{definition}

% where each $\pi^i$ is the policy for agent $i$ as previously defined.
% The policy $\pi^i(s^i_t,
% a^i_t)$ represents the probability of taking action $a^i$ in state $s^i$ at time
% $t$.

A policy $\pi^i$ for agent $i$ is defined as a joint probability distribution
over the set of states and action, i.e., $\pi^i: \mathcal{S}^i \times
\mathcal{A}^i \to [0, 1]$. A joint policy $\pi$ can be viewed as the tuple of
all individual agent policies: $\pi = (\pi^1, \pi^2, \ldots, \pi^n)$. A
trajectory $\tau^i$ of agent $i$ induced by policy $\pi^i$ is defined as a
$(T+1)$-length sequence of state-action pairs: 
$\tau^i = \{
    (s_0^i, a_0^i ), (s_1^i, a_1^i ), \ldots, (s_T^i, a_T^i )\}, \notag \text{where, } \forall t < T: a_t^i \sim \pi^i(a^i | s = s_t^i),
    (s_{t+1}^i, r_{t+1}^i) \sim \mathcal{T}.$
\noindent A trajectory $\tau$ for all agents in the game is a tuple of individual agent
trajectories: $\tau = (\tau^1, \tau^2, \ldots, \tau^n)$, where each $\tau^i$ is
as defined previously for agent $i$. Let $\eta^i$ represent the expected sum of
rewards (return) for agent $i$ under policy $\pi^i$,
and \(\eta\) the total return across all agents under the joint policy $\pi$:  
\begin{equation*}
    \begin{array}{cc}
       \eta^{i} = \mathbb{E}\left[\sum_{t=0}^T \gamma^t r_t^i \mid s_0^i, \pi^i\right]
       ,&\qquad
       \eta = \sum_{i \in \mathcal{N}} \eta^i
    \end{array}
\end{equation*}

Here, $s_0^i$ is the initial state of agent $i$.
An optimal policy $\pi^*$ is defined as the policy that maximizes the total
expected reward $\eta$ across all agents in the system. The optimal policy
$\Pi^*$ is composed of individual policies $\Pi^* = ({\pi^{1}}^*, {\pi^{2}}^*,
\ldots, {\pi^{n}}^*)$ for each agent, and \( \eta(\Pi^*) \) is the optimal joint policy consisting of individual joint policies. We note that in our setting, we use
homogeneous agents with identical tasks, so instead of using per-agent policies, we have $n$
instances of the same policy $\pi^*$.
 
In our setting, we assume that at time $t=0$, agents can start at some location
defined by the set of initial states $\initialstates$, must avoid obstacles
defined by the set of locations in $\obstacles$ at all times, and must reach a
goal in the set of goal states: $\goalset$ within time $T$ while avoiding obstacle
states at all times. Formally, $s^i_0 \in \initialstates, s^i_t \notin
\obstacles \text{ and } s^i_T \in \goalset, \ \forall i \in \mathcal{N} \text{ and }  \ 0 \leq
t \leq T$.
% \begin{todo}
%     {Reviewer 1 has asked for a citation for this defintion and also says $\sigma$ is not defined in Definition 3}
% \end{todo}
\paragraph{Weighted Automata (WA) based rewards:}
We now introduce \emph{weighted automata} (WA), and show how they can be used to define rewards.
In general, weighted automata are finite-state machines that have a notion of
\emph{accepting} conditions and, for each transition in the automaton, the
machine outputs a \emph{weight} for the transition \cite{droste2009weighted}.
When interpreted over algebraic \emph{semirings}, the weights in the automaton
can be used to define various different types of quantitative objectives along
with the acceptance condition.

We eschew the general definition of weighted automata, and look at the 
specific instances of \emph{quantitative languages}
\cite{chatterjee2008quantitative,boker2021quantitative,chatterjee2016quantitativea}, where
the weights of the transitions in the automaton are interpreted over \emph{valuation 
functions}.
This formalism is a generalization of the recently popular \emph{reward machine} approach
used in reinforcement learning frameworks with temporally dependent tasks~\cite{CamachoEtAl2018,icarte2018using,icarte2022reward,zhou2022hierarchical}. 
Weighted automata are more expressive than reward machines as they allow a notion of acceptance in the automaton. We will later empirically show how specific weighted automata can generalize discrete-time temporal logic
objectives (such as Signal Temporal Logic).

\begin{definition}[Weighted Automata]
    A weighted automaton is defined as a tuple \mbox{\(\Ac = \Tuple{Q, Q_I, Q_F, \Sigma, T, \lambda}\)}, where:
    \begin{itemize}
        \item \(Q\) is a finite set of locations, and \(Q_I \subseteq Q\) and \(Q_F \subseteq Q\) are respectively the set of initial locations and final locations;
        \item \(\Sigma\) is an input alphabet;
        \item \(T: Q \times \Sigma \to 2^Q\) is a (partial) labeled transition function, where \(2^Q\) is the powerset of \(Q\);
        \item \(\lambda: Q \times \Sigma \times Q \to \Re\) is a \emph{weight} function.
    \end{itemize}
\end{definition}
% We use \(\Sigma^*\) to refer to the set of all finite length
% sequence of elements \(s \in \Sigma\).

An automaton is \emph{complete} if for all \(q \in Q\) and \(s \in \Sigma\),
there is at least one successor location in the automaton, i.e., \(\abs{T(q, s)}
> 0\). Likewise, the automaton is \emph{deterministic} if for all \(q \in Q\)
and \(s \in \Sigma\), there is exactly one successor location, i.e., \(\abs{T(q,
s)} = 1\). As we intend to use WA-based reward functions, the set of input symbols for the WA is essentially the set of states $\mathcal{S}$ of a stochastic game as defined before.
Thus, given an input trajectory, \(\tau = (s_0, s_1, \ldots, s_l) \in
\Sigma^*\), a \emph{run} in the automaton \(\Ac\) is a sequence of locations
\((q_0, q_1, \ldots, q_{l+1})\) such that \(q_{i+1} \in T(q_i, s_i)\), for \(i
\in 0,\ldots,l\). We use \(\Run_\Ac(\tau)\) to denote the set of runs induced in
\(\Ac\) by \(\tau \in \Sigma^*\). Note that if \(\abs{\Run_\Ac(\tau)} > 0\), the
automaton is non-deterministic, and the run is \emph{accepting} if the last
location \(q_{l+1}\) in the run is in \(Q_F\).
Moreover, for a given input \(\tau = (s_0, s_1, \ldots, s_T) \in \Sigma^*\) and a 
corresponding run \(\sigma = (q_{0}, q_1, \ldots, q_l, q_{T+1})\), a sequence of weights \((w_0, w_1, \ldots, w_T)\) is produced such that \(w_i = \lambda(q_i, s_i, q_{i+1})\).
For some \(\tau \in \Sigma^*\) and a corresponding run \(r \in Q^*\) in the automaton,
we let \(\Weights_\Ac(\tau, r) \in \Re^*\) denote the sequence of weights induced by the run.

\begin{definition}[Valuation function]
    Given the weights \(w \in \Weights_\Ac(\tau, r)\) corresponding to a run \(r \in Q^*\) induced by an input \(\tau \in \Sigma^*\) on a weighted automaton \(\Ac\), a \emph{valuation function} \(\Val(w) \in \Re\) is a scalar function that outputs the \emph{weight} of the sequence.
    The weight of a trajectory \(\tau \in \Sigma^*\) in an automaton \(\Ac\) can be defined as \[w_\Ac(\tau) = \max_{r \in \Run(\tau)} \Val(\Weights(\tau, r)).\]
\end{definition}

For a sequence of weights \(w = (w_0, w_1, \ldots, w_l)\), examples of valuation functions include:
\begin{tasks}[style=itemize](2)
    \task $\mathsf{Sum}(w) = \sum_{t = 0}^l w_t$
    \task \(\mathsf{Avg}(w) = \frac{1}{l + 1} \sum_{t = 0}^l w_t\)
    \task* For a discount factor \(\gamma \in [0, 1]\), \(\mathsf{DiscountedSum}_\gamma(w) = \sum_{t = 0}^l \gamma^t w_t\)
\end{tasks}

In this paper, we are particularly interested in the \(\mathsf{DiscountedSum}_\gamma\) valuation function defined on deterministic weighted automata, as it captures the total discounted rewards semantics for the policy synthesis problem described above.

\revision{\begin{remark}
While we restrict ourselves to deterministic automata for brevity, one can extend the presented framework to non-deterministic automata by resolving the non-deterministic transitions using the \(\max_{}\) operation (or the general semiring addition) \cite{mohri2009weighted,boker2021quantitative}. Moreover, since our framework evaluates the weight of trajectories at the end of \emph{episodes}, one can resolve the non-deterministic runs using the weighting function \(w_\Ac\) as described above.
\end{remark}}

\paragraph{Problem definition:} A specific setting for multi-agent pathfinding that
we wish to solve in this paper is {\em reach-avoid} specifications, i.e., where
multiple agents are each randomly assigned start and goal locations in a shared
environment. Their task is to navigate to their goals as quickly as possible
while avoiding collisions with one another. We can express these conditions as: 
\begin{enumerate}
    \item \textbf{Expeditious goal achievement:} Each agent must reach the goal state within $T$ time-steps, and as expeditiously as possible.
    \item \textbf{Collision avoidance:} Every agent must avoid collisions with other agents and obstacles at all times.
\end{enumerate}

\paragraph{Solution outline:}
Our goal, is to define a reward
function for each agent in the system such that by maximizing the expected total
reward across all agents (denoted by \(\eta\) above), the resulting policy
\(\Pi^*\) achieves the high-level reach-avoid task in a performant fashion.
We do this by defining a weighted automaton that specifies the task for each
agent in the system, i.e., each agent aims to maximize the weight of its
trajectory in the system with respect to the same automaton specification. The
weighted automaton describes a pattern such that maximizing the composed weights
of all agent trajectories, or the \emph{utility} of the system.

%% file: solution_approach.tex
To ensure a timely multi-agent reach-avoid system, we propose a reward structure that emphasizes the importance of efficient trajectories to the goal. Previously in literature, AvSTL \cite{avstl} was developed to assess behavior over time by measuring the area under the signal curve within a given timeframe. Yet, for reach-avoid tasks where goals must be reached quickly and safely, AvSTL may inadvertently favor less direct routes due to its integral-based evaluation. Consider the robustness metric for reaching a goal within time $T$, denoted by $\Ev_{[0,T]} (reachgoal)$, as 1 for success and 0 otherwise. As shown in Figure:~\ref{fig:Counter_example}, an agent that swiftly reaches the goal and moves away may have a lower cumulative value than one following a longer route that lingers at the goal, highlighting a potential discrepancy in AvSTL's approach for such tasks.
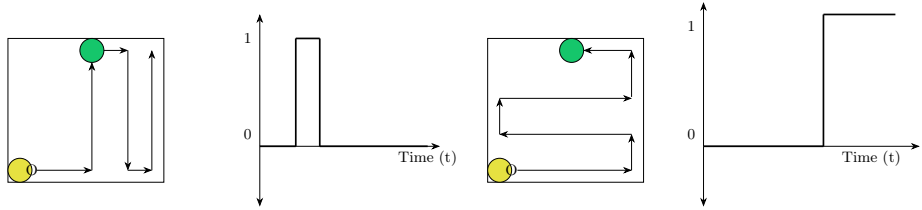
\begin{figure}[!ht]
\centering
\resizebox{1\textwidth}{!}{%
\begin{circuitikz}
\draw  (3,13.75) rectangle (6.25,10.75);
\draw  (13,13.75) rectangle (16.25,10.75);
\draw [ fill={rgb,255:red,230; green,225; blue,65} ] (3.25,11) circle (0.25cm);
\draw [ fill={rgb,255:red,230; green,225; blue,65} ] (13.25,11) circle (0.25cm);
\draw [ fill={rgb,255:red,21; green,198; blue,107} ] (4.75,13.5) circle (0.25cm);
\draw [ fill={rgb,255:red,21; green,198; blue,107} ] (14.75,13.5) circle (0.25cm);
\draw [->, >=Stealth, solid] (3.6,11) -- (4.75,11);
\draw [->, >=Stealth, solid] (4.75,11) -- (4.75,13.25);
\draw [->, >=Stealth, solid] (5,13.5) -- (5.5,13.5);
\draw [->, >=Stealth, solid] (13.62,11) -- (16,11);
\draw [->, >=Stealth, solid] (16,11) -- (16,11.75);
\draw [->, >=Stealth, solid] (16,11.75) -- (13.25,11.75);
\draw [->, >=Stealth, solid] (13.25,11.75) -- (13.25,12.5);
\draw [->, >=Stealth, solid] (13.25,12.5) -- (16,12.5);
\draw [->, >=Stealth, solid] (16,12.53) -- (16,13.5);
\draw [->, >=Stealth, solid] (16,13.5) -- (15,13.5);
\draw [ color={rgb,255:red,4; green,0; blue,255} , line width=1.3pt , solid] (3.25,11) circle (0cm);
\draw [ fill={rgb,255:red,22; green,24; blue,23} ] (9.75,13.5) circle (0cm);
\node [font=\small] at (3.5,11) {O};
\node [font=\small] at (13.5,11) {O};
\draw [->, >=Stealth, solid] (5.5,13.5) -- (5.5,11);
\draw [->, >=Stealth, solid] (5.5,11) -- (6,11);
\draw [->, >=Stealth, solid] (6,11) -- (6,13.5);
\draw [->, >=Stealth] (8.25,11.5) -- (8.25,14.25);
\draw [->, >=Stealth] (17.5,11.5) -- (17.5,14.5);
\draw [->, >=Stealth] (8.25,11.5) -- (12,11.5);
\draw [->, >=Stealth] (17.5,11.5) -- (22,11.5);
\draw [->, >=Stealth] (8.25,11.5) -- (8.25,10.25);
\draw [->, >=Stealth] (17.5,11.5) -- (17.5,10.25);
\draw [line width=1pt, short] (8.25,11.5) -- (9,11.5);
\draw [line width=1pt, short] (9,11.5) -- (9,13.75);
\draw [line width=1pt, short] (9,13.75) -- (9.5,13.75);
\draw [line width=1pt, short] (9.5,13.75) -- (9.5,11.5);
\draw [line width=1pt, short] (9.5,11.5) -- (11.75,11.5);
\draw [line width=1pt, short] (17.5,11.5) -- (20,11.5);
\draw [line width=1pt, short] (20,11.5) -- (20,14.25);
\draw [line width=1pt, short] (20,14.25) -- (21.5,14.25);
\node [font=\small] at (11.75,11.25) {Time (t)};
\node [font=\small] at (21,11.25) {Time (t)};
\node [font=\small] at (8,11.75) {0};
\node [font=\small] at (17.25,11.75) {0};
\node [font=\small] at (8,13.75) {1};
\node [font=\small] at (17.25,14) {1};
\end{circuitikz}
}%
\caption{\textit{Counter-example: An agent (O) tries to go from the initial location (yellow) to the goal location (green). Under general expeditious semantics (such as in AvSTL) the trajectory in the left figure where the agent reaches the goal quickly but wanders after would receive a lower reward (area under the curve) than the one in the right figure where the agent takes a longer path to reach the goal but stays in the goal for longer.}}
\label{fig:Counter_example}
\vspace{-10pt}
\end{figure}

% \subsection{Reward Shaping}

In this section, we outline our reward-shaping mechanism which is designed to promote efficient goal attainment with collision avoidance.
 The reward function provides a large reward $(+b)$ for reaching the goal as sufficient incentive for doing so, while penalizing each step \textit{before} reaching the goal with a minor negative reward $(-a)$ to hasten the process.
% \begin{todo}
%     {Comment on DFA/NFA - whatever Nicola said}
% \end{todo}
\begin{figure}
    \centering
    % \vspace{-10pt}
    \includegraphics[width=0.95\linewidth]{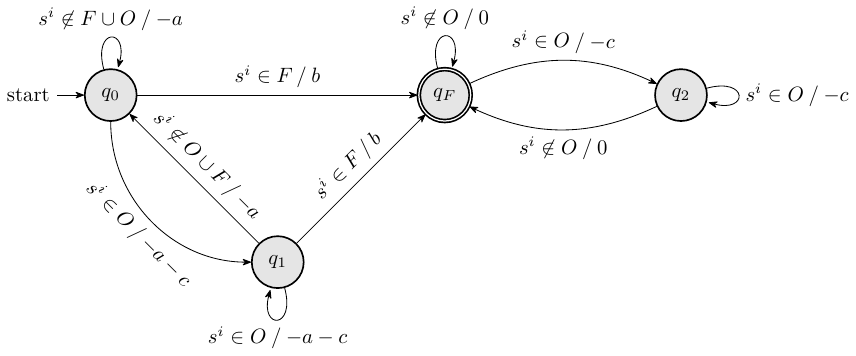} % Slightly less than the wrapfig width to account for padding
    \caption{\textit{A deterministic weighted automaton defining the reach-avoid task that needs to be completed by each agent in the system. In the figure, \(s^i\) refers to the current state of agent \(i\) in the multi-agent system.}}
    \label{fig:weighted-task-automaton}
    \vspace{-10pt}
\end{figure}
 It also generates a hefty penalty for collisions $(-c)$ at any point which could be smaller than or equal to the reward for reaching the goal, depending on whether or not the system wants to permit trajectories with collisions or not. This structured reward shaping aims to balance the urgency of completing tasks with the imperative of maintaining safe operations within the multi-agent system. 
We abuse notation and denote $\traj^i(t)$ to be the state of the agent $i$ at time $t$ in the trajectory. The reward function $R$ for an agent $i$ at time $t$ taking action $a$ resulting in a state transition from $s^i_t$ to $s^{i}_{t+1}$ can be defined as:
\begin{equation}
R^i_t(s^i_t, a, s^{i}_{t+1}) = f_F(\traj^i(t))  + f_O (\traj^i(t)) 
\end{equation}
Now we can define:
\begin{equation}
    \begin{array}{cc@{,\ }cc}
       f_O(\traj^i(t))  =&  \begin{cases}
-c &~\text{if}~ s^{i}_{t} \in \obstacles \\
0  &~\text{if}~ s^{i}_{t} \notin \obstacles
\end{cases} 
&
       f_F(\traj^i(t))  =& \begin{cases}
0 &~\text{if}~\exists t' < t: s^{i}_{t'} \in \goalset \\
b  &~\text{if}~ s^{i}_{t} \in \goalset ~\text{\&}~ \nexists t' < t: s^{i}_{t'} \in \goalset\\
-a &~\text{if}~ s^{i}_{t} \notin \goalset ~\text{\&}~ \nexists t' < t: s^{i}_{t'} \in \goalset
\end{cases}
    \end{array}
\end{equation}

% \end{equation}
where $a,b,c \in \mathbb{R}^+$; $s^{i}_{t} \in \goalset$, indicates agent $i$ reaching the goal and $s^{i}_{t} \in \obstacles$ indicates a collision with an obstacle. 
We also define the constraint: $b \ge c > a \cdot T$, where $T$ is the total time permitted to reach the goal (i.e., the length of an episode in terms of timesteps). This
provides a hierarchical framework where reaching the goal is paramount, followed by the penalty for collisions, with a smaller penalty for taking steps before reaching the goal.  

The above reward structure generally follows the weights from the automaton shown in
Figure~\ref{fig:weighted-task-automaton}.
Moreover, for a discount factor \(\gamma \in (0,1)\) we can interpret the trace over the $\mathsf{DiscountedSum}_\gamma$ valuation function.

Let us denote $toa(\traj)$ as the time of arrival of the trajectory into the goal, where $toa(\traj) = \min \{t \mid (s_t, a_t) \in \traj  \ \text{and} \  s_t \in \ F \}$.
\begin{proposition}
If the sum of rewards over a trajectory is positive, the trajectory satisfies the condition of reaching the goal within \(T\) timesteps, and not colliding with an obstacle.
    % \( \sum_{t=0}^{T}R^i_t(s^i_t, a, s^{i}_{t+1}) > 0 \implies \mathcal{F}_{[0,T]} \text{goal} \wedge \mathcal{G}_{[0,T]} (\text{not obstacle}) \)
    \(w_\Ac(\tau) > 0 \implies \Ev_{[0,T]} (s_{toa} \in \ F) \wedge \Alw_{[0,T]} (\neg s \in \obstacles)\)
    \footnote{\(\Ev_{[a,b]} (x > 0)\)
    denotes \(\exists t \in [a,b] ~\text{s.t.}~  x(t) > 0\) and \(\Alw_{[a,b]} (x > 0)\) denotes \(\forall t \in [a,b] ~\text{s.t.}~ x(t) > 0\).}
\end{proposition}
All timesteps after $toa(\traj)$ are weighted $0$ according to our reward function, as we assume the absence of collisions, so we only consider the part upto $toa(\traj)$.
\begin{lemma}
    Trajectories that are more expeditious i.e. that reach the goal sooner (assuming they do not have collisions) have higher rewards returned by the weighted automata.  
     \[toa(\traj) < toa(\traj') \Rightarrow w_A(\traj) >  w_A(\traj') \] 
\end{lemma} 

\begin{proof}
    The proof follows from the definitions of the reward function and $w_A(\traj)$.
\end{proof}

\begin{comment}
\begin{theorem}
All trajectories reaching the goal are of the type \((q_0^* \mid q_1^*) q_F\).
\todo{Why do we need this? This is the definition of an accepting run...}
\end{theorem}

\begin{proof}
The proof follows from the fact that of all paths in the weighted automaton (WA) that reach the goal state, the shortest path corresponds to the smallest time of arrival $(toa)$, and has the most positive return. This optimality condition ensures that any trajectory reaching the goal can be expressed in the form of \((q_0^* \mid q_1^*) q_F\), where \(q_0^*\) and \(q_1^*\) represent sequences of states leading optimally to the final state \(q_F\).
\end{proof}
\end{comment}

%% file: egt_framework.tex
% \newcommand{\expected}{\expected}

% \paragraph{Traditionally used Approaches:}
In our multi-agent system, we employ homogeneous agents that are centrally trained under a unified framework, allowing for the sharing of a joint policy. This homogeneity simplifies the training process and enables the application of single-agent methodologies, such as $A^*$ and single-agent reinforcement learning algorithms to manage the collective behavior. 

\paragraph{Search-based methods for policy learning:} Heuristic-based approaches to pathfinding, exemplified by algorithms such as $A^*$ \cite{astar}, D-star \cite{dstar}, and D-star Lite \cite{dstarlite}, utilize heuristics to navigate efficiently from one point to another within an environment. These algorithms operate by employing a heuristic to estimate the cost from any node in the search space to the goal, effectively guiding the search process toward the most promising paths while minimizing unnecessary exploration. Common heuristics include the Manhattan distance, which provides a direct estimation of the minimal possible distance to the goal, assuming a grid-like path with no obstructions. This heuristic foreknowledge is crucial as it significantly influences the efficiency and effectiveness of the search. 

\paragraph{Learning-based methods:} Data-driven policy optimization methods are
increasingly popular in environments where the system is stochastic, the model
of the system is not available, and when the state-space is large. Model-free reinforcement learning methods such as Q-learning \cite{watkins}, Deep Q-Networks (DQN) \cite{dqn}, Advantage Actor Critic (A2C) \cite{a2c}, and Proximal Policy Optimization (PPO) \cite{ppo} refine their strategies through extensive interaction with the environment. 

Q-learning is a foundational off-policy algorithm where each pair of state $s$ and action $a$ is associated with a $Q(s,a)$ value representing the expected future reward. This value is iteratively updated using the Bellman equation \cite{watkins}. In multi-agent systems, Q-learning can be adapted by extending the state-action space to include all possible combinations of states and actions for every agent, effectively using a shared Q-table. This approach allows all agents to follow a centrally managed, uniform policy that operates based on the joint action space. Deep Q-Networks (DQN) extends Q-learning by employing deep neural networks to approximate the $Q(s,a)$ function, allowing it to handle continuous and high-dimensional state spaces efficiently. 

A2C is an actor-critic method that uses multiple parallel environments to update its policy and value networks reducing variance and improving learning speed. PPO is a more stable and efficient actor-critic method that optimizes a clipped version of the objective function instead of directly optimizing the objective function to provide smoother updates. Actor-critic methods can be implemented with each agent operating simultaneously in a shared environment, all contributing to a central policy update mechanism. This ensures that learning is synchronized.

Monte Carlo methods, another branch of model-free RL, do not assume knowledge of the environment's dynamics and instead rely on sampling full trajectories to estimate expected returns. These methods can be adapted for multi-agent use by sampling and averaging returns across all agents, using these aggregated insights to update a central policy that guides all agents.

We explored a spectrum of strategies for addressing the multi-agent pathfinding challenge. These included the traditional heuristic-driven approach such as A*, tabular temporal difference methods like Q-learning, neural network-driven techniques exemplified by Proximal Policy Optimization (PPO), and stochastic sampling-based approaches such as those used in Monte Carlo methods.

\paragraph{Limitations of these approaches:}
Limitations of search-based algorithms include scalability as the frontier grows exponentially with an increase in the number of agents, lack of applicability in stochastic and dynamic environments, and a priori knowledge of the environment in terms of an admissible heuristic which we {\em do not provide} to any of the other algorithms including our own. 
The primary shortcoming of learning-based algorithms such as the ones we described above that we aim to overcome is sample-inefficiency. These algorithms can require a large number of interactions with the environment to learn effective policies, which is impractical in complex or time-sensitive applications. These algorithms also treat each step or trajectory as equally important for a policy update. However, we aim to replace that with a weighted update that ensures that steps in trajectories that overperform or underperform by a significant magnitude, also make a more significant impact to the policy.

\subsection{Algorithmic framework using Evolutionary Game Theory} 
{\em Evolutionary Game Theory} (EGT) presents a dynamic alternative to
classical game theory and is better suited to the realities of multi-agent
systems.  \cite{egt_stanford,egt_games,egt_sandholm} Unlike traditional game
theory, it does not presume rationality among agents, making it adaptable to any
scenario. It focuses on the dynamics of strategy changes driven by the success
of current strategies in the population, reflecting a more naturalistic approach
to agent learning. Evolutionary strategies also offer quicker and more robust
convergence guarantees, making them particularly appealing in environments where
reinforcement learning is an intuitive fit. Given these advantages, it is a
promising direction for advancing research and practical applications in
multi-agent settings.

\paragraph{Replicator equation:}
The key concept within EGT pertinent to us is that of {\em replicator dynamics}, which describes how the frequency of strategies (or policies in RL) changes over time based on their relative performance.
The classic replicator equation in evolutionary game theory describes how the proportion of a population adopting a certain strategy evolves over time. Mathematically, it is expressed as \cite{replicator}:
\begin{equation}
    x_j(i+1) = x_j(i) \cdot \frac{f_j(i)}{\bar{f}(i)} \label{eq:og_rep}
\end{equation}
where $x_j(i)$ represents the proportion of the population using strategy $j$ at time $i$, $f_j(i)$ is the fitness of strategy $j$, and $\bar{f}(i)$ is the average fitness of all strategies at time $i$. The equation indicates that the growth rate of a strategy's proportion is proportional to how much its fitness exceeds the average fitness, leading to an increase in the frequency of strategies that perform better than average.

\paragraph{Representation of populations and the fitness  equivalent:} We represent the probability distribution over actions in a given state as a population. Each individual in the population corresponds to an action, and the proportion of (individuals corresponding to a specific action) in the  population represents the probability of taking the action in that state. Thus, the population represents a stochastic policy.
The fitness function measures the reproductive success of strategies based on payoffs from interactions, similar to utility in classical game theory, or how in RL, the expected return measures the long-term benefits of actions based on received rewards. Both serve as optimization criteria: strategies or policies are chosen to maximize these cumulative success measures to guide them towards optimal behavior.
Therefore, in our model, the fitness for a state \( f(s) \) corresponds to the expected return from that state, equivalent to the value function \( v(s) \), and the fitness for a state-action pair \( f(s,a) \) corresponds to the expected return from taking action \( a \) in state \( s \), equivalent to the action-value function \( q(s,a) \). 

Under the assumption of sparse rewards— where significant rewards are received only upon reaching specific states or goals, \( f(s) \) is defined as  \( \EX[f(\tau_s)] \), the expected return across all trajectories through state \( s \). Likewise, \( f(s,a) \) is defined as \( \EX[f(\tau_{(s,a)})] \), the expected return across trajectories involving the state-action pair \( (s,a) \). 
% % \begin{align}
% %     q(s,a) &= f(s,a) \approx \expected[f(\tau_{(s,a)})],\label{eq:q_sparse} \\
% %     v(s)   &= f(s) \approx \expected[f(\tau_{(s)})] \label{eq:v_sparse}
% % \end{align}

\paragraph{Policy update mechanism:} The replicator equation can be adapted to update the probability of selecting certain actions based on their relative performance compared to the average. The adaptation of the replicator equation is as follows:
\begin{equation}
\label{eq:policyupdate}
    \pi^{i+1} (s,a) = \frac{\pi^i (s,a) f(s,a)}{\sum_{a' \in A} \pi^i (s,a') f(s,a')} 
\end{equation}
where $\policy^{i}(\agentstate, \action) $ is the  probability of action $\action$ in state $s$, in the $i^{th}$ iteration, and $\pi^{i+1}$ represents the policy in the ${i+1}^{th}$ iteration. 

\paragraph{\revision{Heterogeneous agents:}}
\revision{We note that in our approach agents are assumed to be homogeneous, i.e., the 
learned policy is shared across all the agents. Our approach can theoretically work
for heterogeneous agents as well. However, each agent is then
required to maintain its own policy. This means that 
for each agent, per agent state, a population of actions will
require to be updated.}

\paragraph{Prioritizing expeditious and safe trajectories:}
The adaptation of our policy update equation focuses on the ratio of the expected fitness of a particular action to the average expected fitness, allowing for a more nuanced update. It scales the probability of each action relative to how much better (or worse) it performs compared to the average, rather than simply whether it is better or worse. 
Actions that lead to higher returns relative to the average are thus promoted in subsequent iterations of the policy. This selective pressure inherently favors actions that contribute to reaching goals faster and avoiding collisions, as they have a greater impact on the expected return due to the structured rewards. As a result, the policy evolves towards a strategy that seeks to maximize rewards by combining efficiency with safety. The learning rate \( \alpha \) ensures that the update remains bounded and allows for fine-tuning of the learning process.
\renewcommand{\tcp}[1]{\textcolor{blue}{// #1}}

 \begin{figure}[t]
    \centering
    % \vspace{-10pt}
    \includegraphics[width=1\linewidth, trim=2.55cm 1.5cm 5.25cm 0cm, clip]{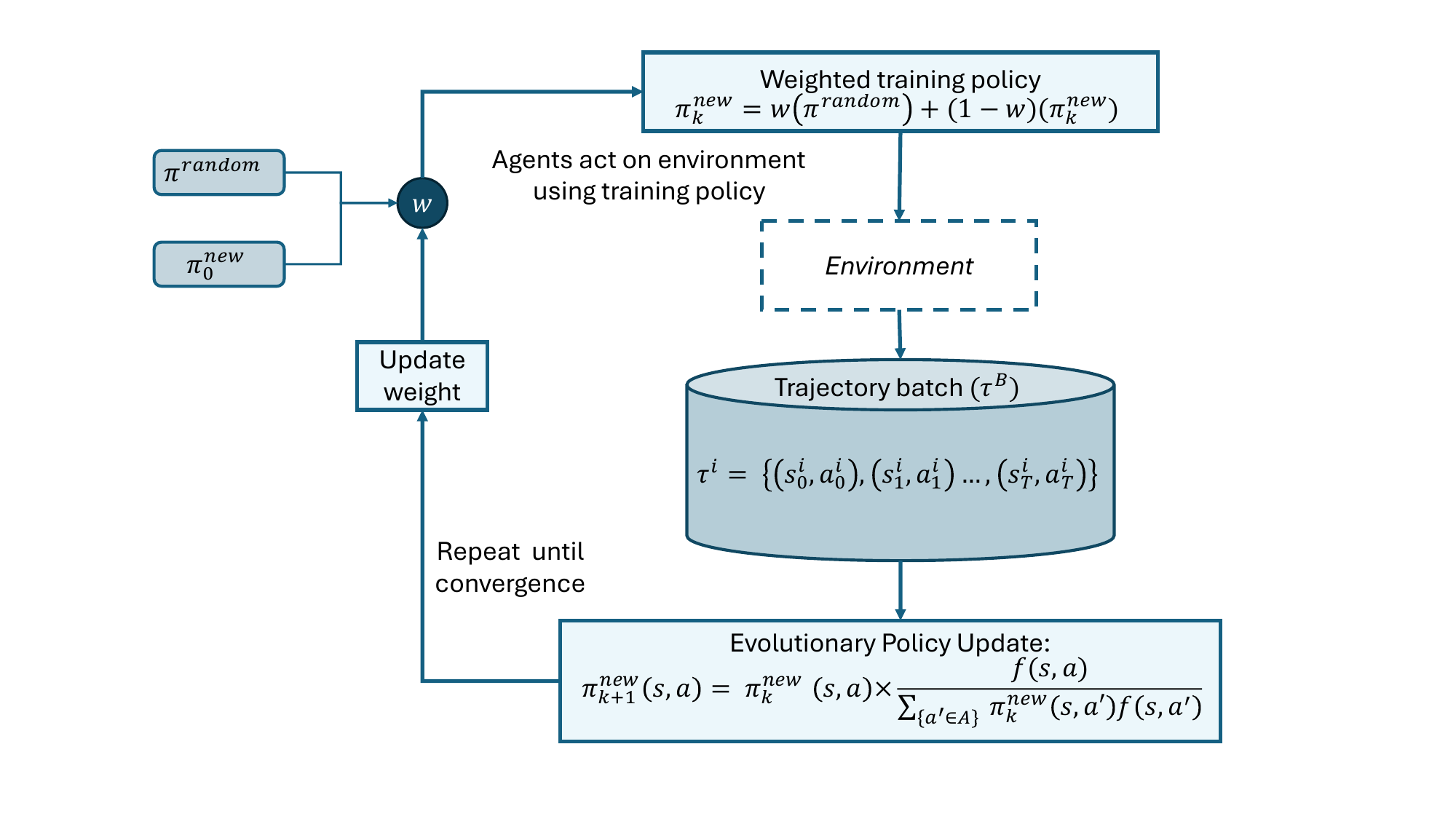}
    \caption{ \revision{\textit{Overview of the evolutionary based learning approach used in our algorithms.}}}
    \label{fig:policy_population}
    \vspace{-10pt}
\end{figure}

\DontPrintSemicolon
\begin{algorithm}[t]
\caption{Multi-agent Pathfinding using EGT}\label{alg:optimized_policy} 
% \SetKwInOut{Input}{Input}
% \SetKwInOut{Input}{Input}
\SetKwInput{Input}{Input}
\SetKwInput{Output}{Output}
% \SetAlgoLined
% \Input{Hyper-parameters
% \begin{itemize}[nosep,label=\textbullet]
%      \item \textit{To alter weight:} $ \nu \in (0,1)$
%      \item \textit{To ensure discounted policies:} $\epsilon > 0$
%      \item \textit{To check convergence:} $\delta > 0$ 
%      \item \textit{Scaling factor/learning rate:} $\alpha > 0 $  
% \end{itemize}}
\Input{Hyper-parameters
\begin{itemize}[label=\textbullet]
    \item To alter weight: $\nu \in (0,1)$
    \item To ensure discounted policies: $\epsilon > 0$
    \item To check convergence: $\delta > 0$
    \item Scaling factor/learning rate: $\alpha > 0$
\end{itemize}}
\Output{Trained policy}
$k \gets 1$  \tcp{$k$ is the iteration number} \;
$\eta_0 \gets 0$  \tcp{$\eta_k$ contains the expected at the end of the $k^{th}$ iteration} \;
$\pi^{new}_0 (s,a)$ $\leftarrow$ $\frac{1}{|A|}$, $\forall s\in S, a \in A$ \tcp{Initialize policies as random} \;

% \end{itemize}
\While{True}{
     \For{each episode in batch $b = 1$ to $B$}{ \nllabel{algoline:loopstart}
         Generate trajectory $\traj^i_b$ for each agent using $\pi^{new}_k$ \;
          Append the collective generated trajectory $\traj_b$ to batch \;   
     }
      Compute expected return $ \eta_k $ over the batch \;
     \If {$\eta_k - \eta_{k-1} \geq \delta$ }{  \nllabel{algoline:convergence}
         \For{all trajectories $\traj_b$ in batch} { 
             \For{each  state-action pair $(s,a) \text{ in trajectory } \traj_b$}{
             \textbf{Evolutionary Policy Update:} \;
                 ${\policy_{k+1}^{new}(s,a)}$ $\gets$ $\policy_{k}^{new}(s,a)  \cdot \alpha \cdot \frac{\pi^{new}_k (s,a) f(s,a)}{\sum_{a' \in A} \pi^{new}_k (s,a') f(s,a')}  $ \nllabel{algoline:loopend}
             }
         }
     $w \gets max ( \epsilon, w - \nu \ )$ \tcp{Decrement weight on exploratory policy} 
    $\pi^{new}_{k+1} \gets w \cdot \frac{1}{|A|} \ +(1- w)\cdot \pi_{k+1}^{new}$ \tcp{Weighted policy update}
    $k \gets k + 1$}
    \lElse{\Return $\pi^{new}_{k} $  \tcp{Return the final optimized policy}
          \nllabel{algoline:Return}}
}
\end{algorithm}

\paragraph{Proposed algorithm:}
Our algorithm uses batch-based updates: We initialize the policy $\pi^{new}$ to
be initially random. We sample trajectories as part of a batch, and the
state-action pairs in these trajectories are updated according to the update
rule.  The return for the $(k+1)^{th}$ iteration is set as the return for the
current batch of trajectories. We maintain a weighted discounted policy while
training, with the weight decreasing with each iteration. This process is
repeated until our termination condition has been met i.e. $\eta_{k+1} - \eta_k
> \delta$: This condition checks if our policy is improving with each update.
If it does not, we say it has converged.

\paragraph{\revision{Termination and convergence:}} \revision{We can show that the
algorithm terminates after a finite number of iterations, based on the
observation that policy updates ensure that the utility/value of each state
monotonically increases. The evolutionary update guarantees that the probability
$\pi^{new}(s,a)$ increases only when $f(s,a)$ is greater than $\EX_a(f(s,a))$,
while $\pi^{new}(s,a)$ decreases whenever $f(s,a) < \EX_a(f(s,a))$. This in turn
guarantees that under $\pi^{new}$, the value of the state $s$ increases. For
states not sampled in batch, $\pi^{new}(s,a)$ is unchanged. Now there are two
cases: (1) The maximum improvement across all states between two consecutive
iterations is less than $\delta$, in which case the algorithm goes to
Line~\ref{algoline:Return} and returns the policy $\pi^{new}$. (2) The maximum
improvement is greater than $\delta$. In the second case, the algorithm cannot
forever increase the expected return $\eta_k$, as the maximum expected return is 
guaranteed to be finitely upper bounded (as trajectories are finite in length).}

\revision{We note that our algorithm may converge to a policy that is sub-optimal; however,
as long as the return $\eta_k$ is positive, the trained policy satisfies the
weighted automaton objectives. To prove that the algorithm terminates to an
optimal policy would require the policy update operator to be a contraction
mapping, similar to learning algorithms like Q-learning \cite{watkins} or
value iteration \cite{Bert05}; we defer this extension to future work.}

\revision{To analyze the complexity of our algorithm, we note that each loop (i.e. from Lines ~\ref{algoline:loopstart} to ~\ref{algoline:loopend}) runs until $\eta_{k} - \eta_{k-1} < \delta$. If we assume the existence of an optimal policy that maximizes returns, and let the return under the optimal policy be denoted as $\eta^*$, then this loop can run for a maximum of $\frac{\eta^* - \eta_0}{\delta}$ times. Each loop itself has a complexity of $T \cdot B$ (with $T$ being the length of an episode and $B$ being the number of episodes in a batch). Then, the  complexity of our algorithm is $ O \left( \frac{\eta^* - \eta_0}{\delta} \cdot T \cdot B\right)$.}

% (2) policy update is a {\em contractive} mapping
% \cite{Bert05}. As we essentially combine
% two standard results, we explain each of these informally
% instead of a formal proof.}

 % \revision{{\em Monototonic policy improvement.}}
% \revision{
% \revision{\noindent {\em Contractivity of policy update.}}

% \revision{The contraction property must ensure that the distance between successive policy iterations decreases over time, leading to the convergence of an optimal policy. Our policy update rule adjusts the policy based on the fitness value of the state-action pair and normalizes it to maintain a valid probability distribution. In future work, we could demonstrate by measuring the distance between policies in successive iteration (using a metric like Total Variation Distance or KL Divergence), and show that the distance between successive policies reduces, indicating the algorithm can eventually converge under specific circumstances.}
% mypara{\revision{Complexity}}\revisio
% iven the aforementioned intuitions,it could be possible to show the convergence of our algorithm to an optimal solution provided strong assumptions about the reward structure, and without bounds on the number of iterations required to converge. Theoretically, an optimal solution can only be obtained when the difference in the returns between successive iterations is zero i.e. $\eta_{k+1} - \eta_k \ = 0$. However, in practice as with many reinforcement learning algorithms, we relax this constraint to be some $\delta$, giving us a possibly sub-optimal solution.}

%% file: experiments.tex
For our experiments, we define the world as a two dimensional $n \times n $ grid, where each state $s \in S$, is defined as $(x,y)$ where $x$ and $y$ are coordinates. 
The agents are homogeneous and the set of actions available to each agent $A^i$ = \{up, down, left, \ right, \ stay\}.  At the beginning of each episode the agent is assigned a random start location and must reach one of the goal locations by the end of it.
We vary the grid size from $20\times20$ to $200\times200$, and vary the number of agents from 2 to 50.
We benchmark our algorithm against a planning algorithm $A^*$ \cite{astar}, a Monte-Carlo approach \cite{sutton}, tabular Q-learning \cite{sutton} and a Deep RL algorithm PPO \cite{ppo}. We use Manhattan-distance as the heuristic for $A^*$ and the same reward function as our algorithm for the other approaches.
Our experiments were carried out on a laptop with 2.0GHz dual-core Intel Core i5 processor and 16GB RAM.
We use the `stablebaselines3'  \cite{raffin2021stable} implementation of PPO with standard values for the hyperparameters. Our algorithms have been implemented on OpenAI gym \cite{openai}.

\begin{figure}[ht!]
\centering
% \vspace{-20pt}
\subfloat[Time to reach the goal vs. Grid Size]{%
% \begin{minipage}{0.49\textwidth}
\includegraphics[width=0.49\linewidth]{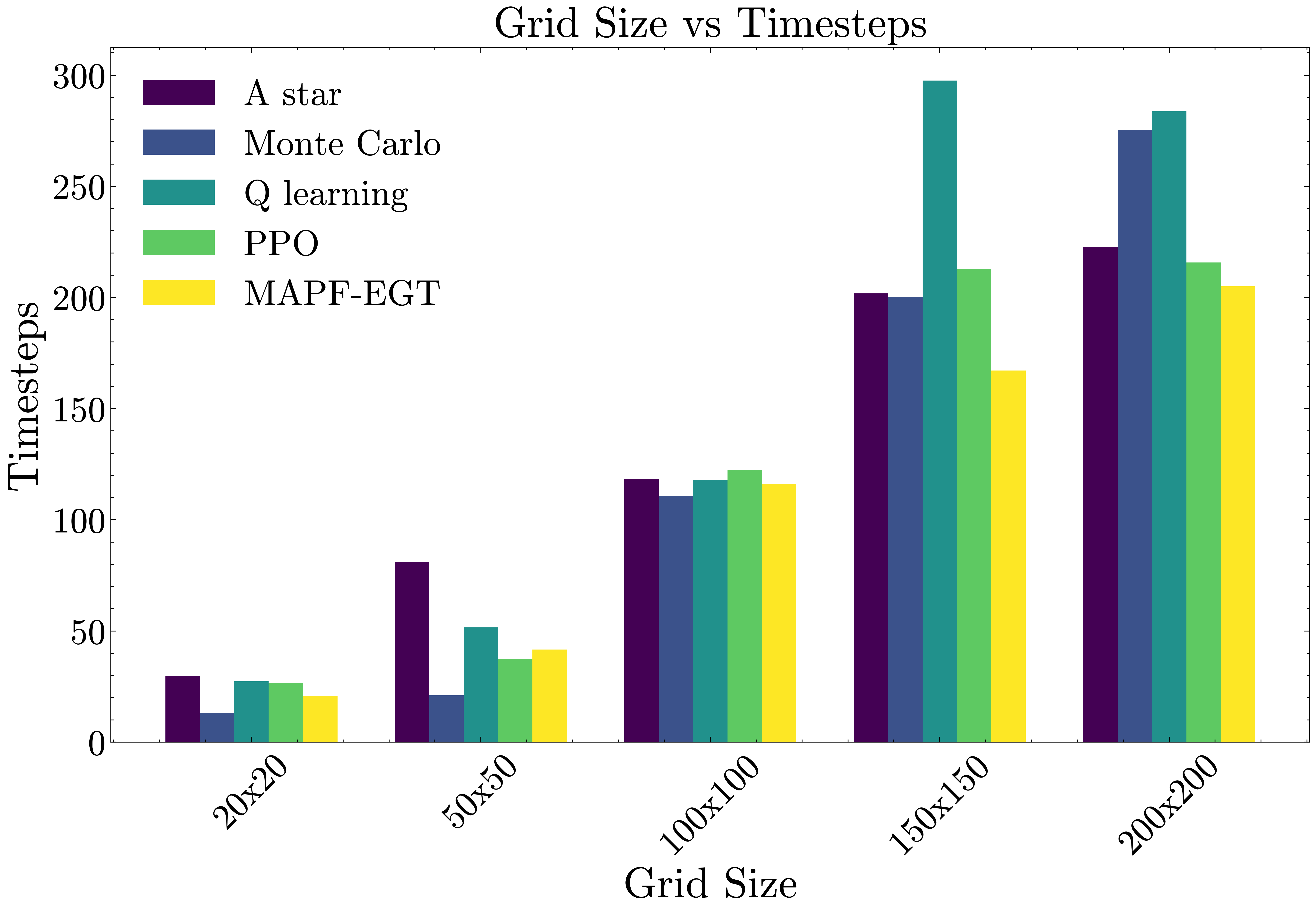}
        \label{fig:grid-timesteps}
}
% \end{minipage}\hfill
\subfloat[Obstacle distance vs. Grid Size]{%
% \begin{minipage}{0.49\textwidth}
\includegraphics[width=0.49\linewidth]{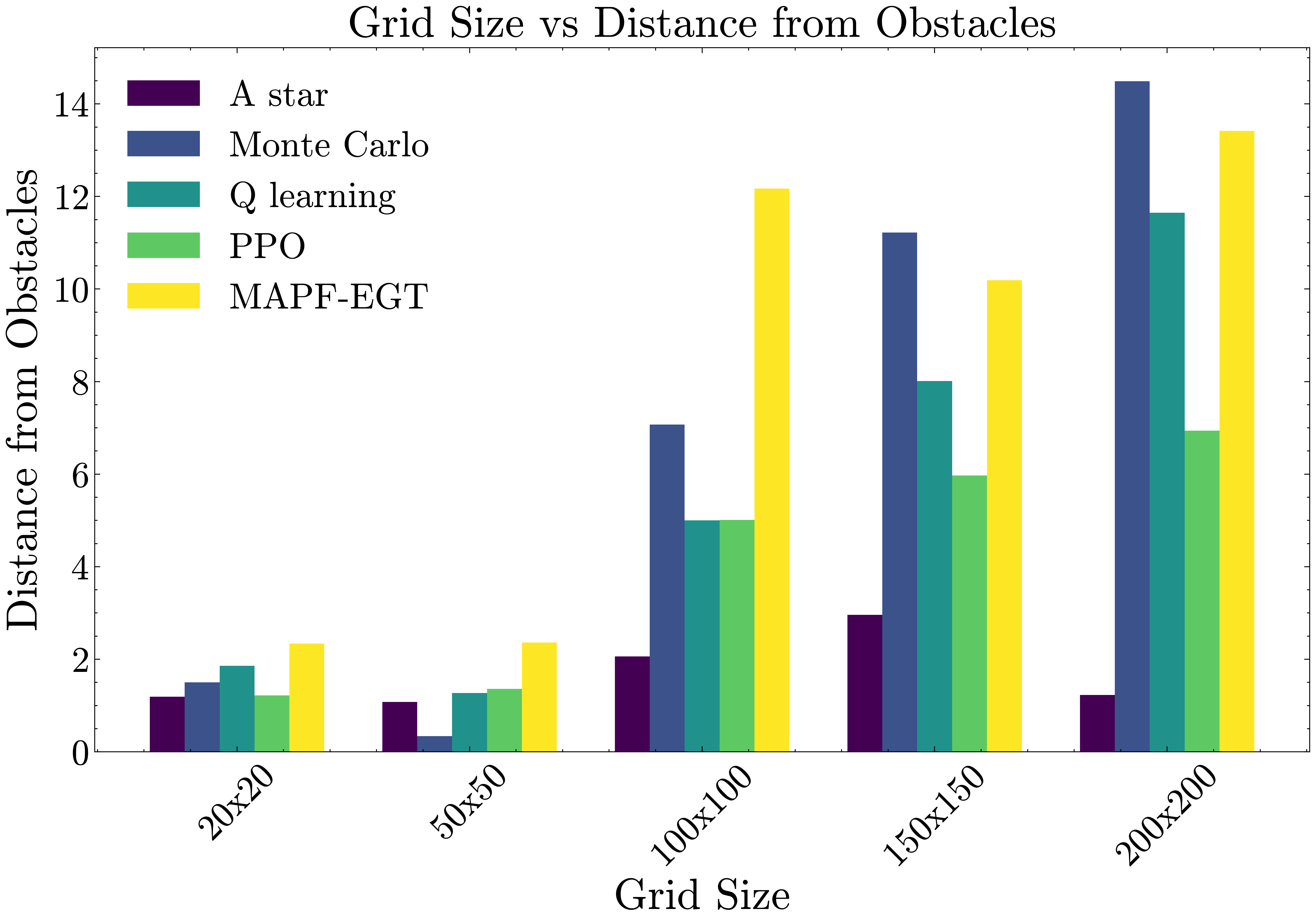} 
% \caption{Obstacle distance vs. Grid Size}
\label{fig:grid-dis}
    % \end{minipage}
} 
\\
% \begin{minipage}{0.49\textwidth}
\subfloat[Clock Time (s) vs Grid Size]{%
\includegraphics[width=0.49\linewidth]{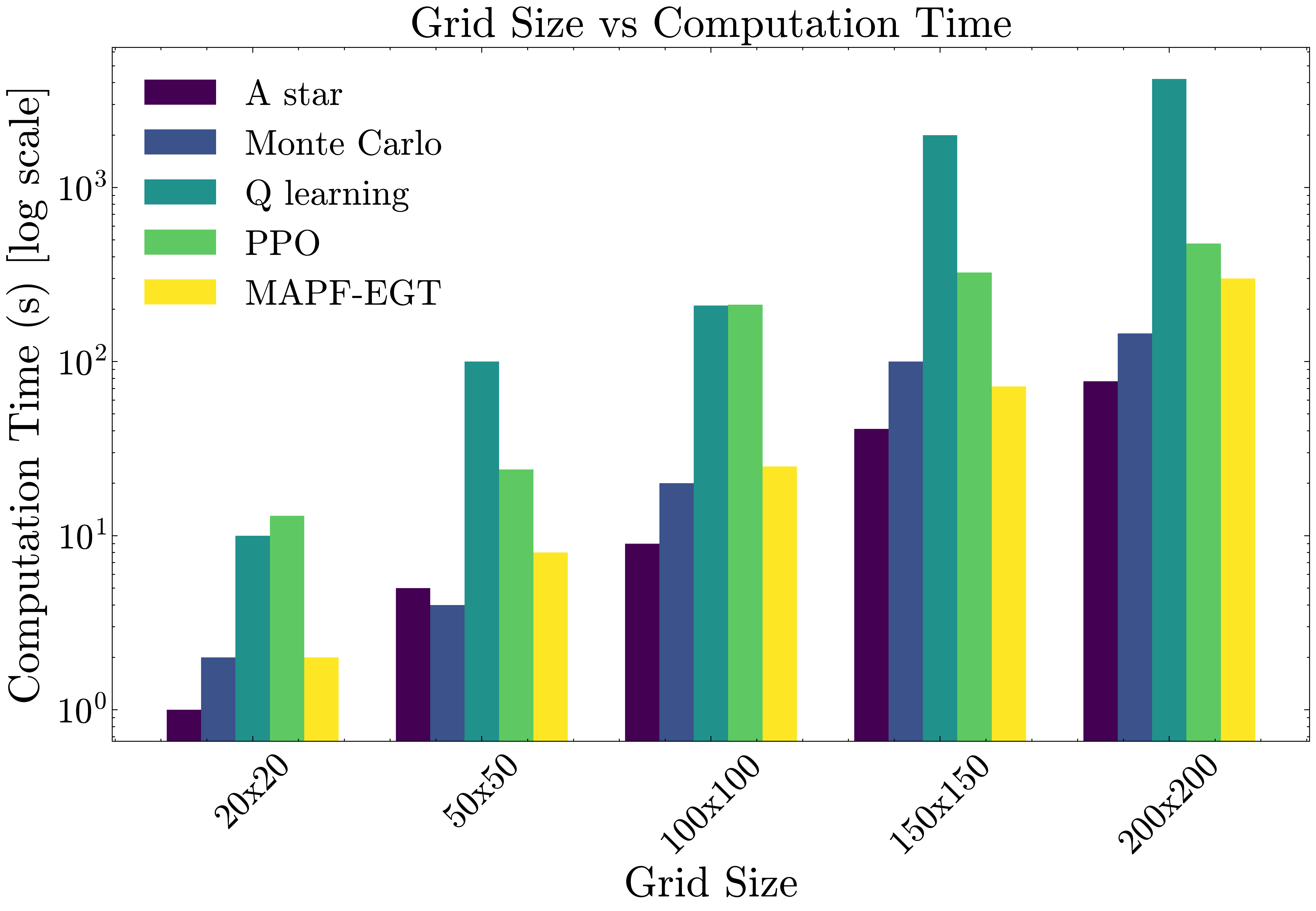}
\label{fig:grid-comp}
}
% \end{minipage}\hfill
\subfloat[Total Time vs. Number of Agents]{%
% \begin{minipage}{0.49\textwidth}
% \includegraphics[width=0.49\linewidth]{FORMATS_Total_Time_(s).png}
\includegraphics[width=0.49\linewidth]{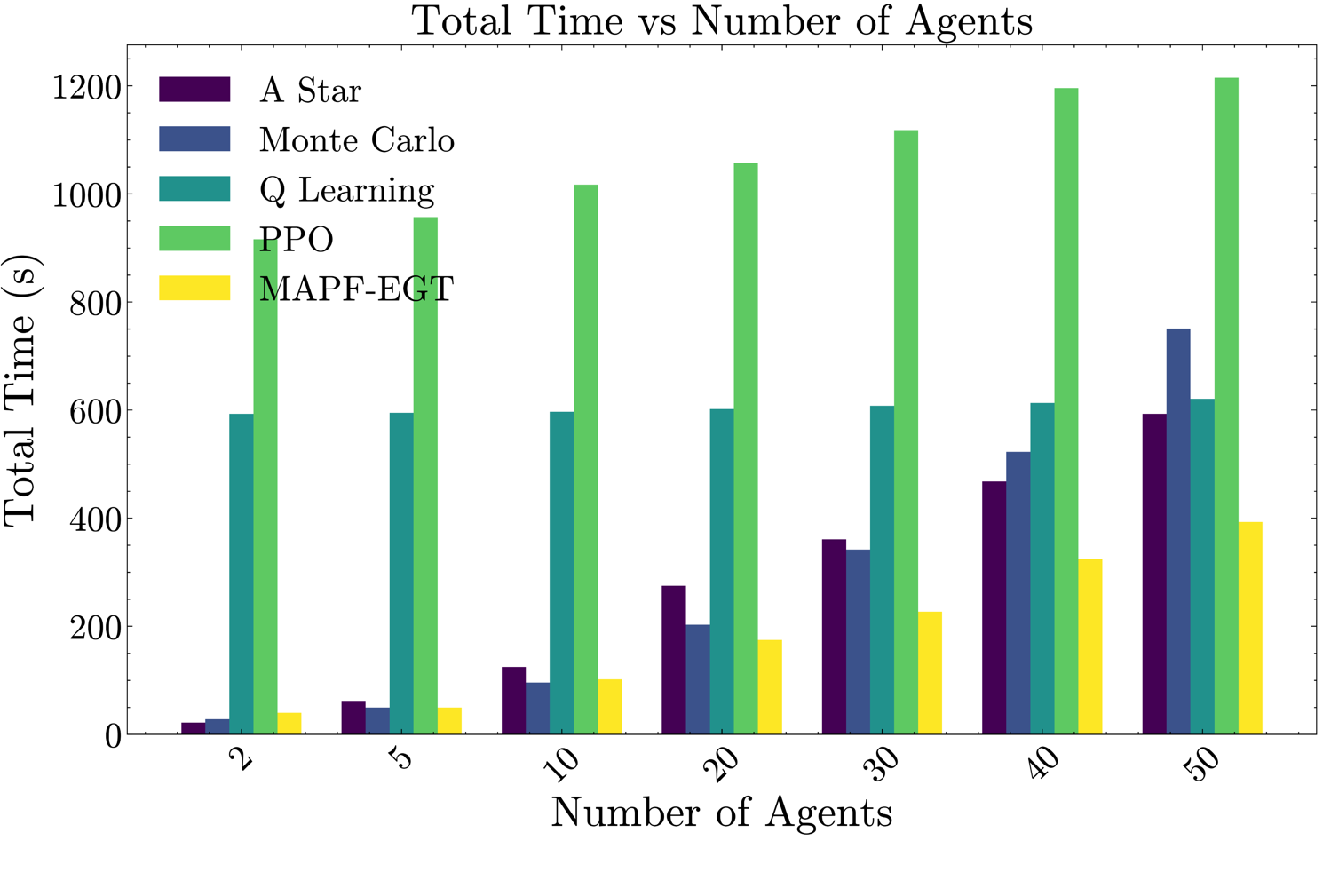}
        \label{fig:total-time-agents}
}
\caption{\textit{MAPF-EGT benchmarked against the algorithms: $A^*$, Monte-Carlo search, PPO, and Q learning. Timesteps to reach the goaal (Fig. a), Expected minimum distance from obstacles (greater distance indicates safer paths)(Fig. b), and clock time (seconds) required for computation (Fig. c) compared across the grid sizes: 20$\times$20, 50$\times$50, 100$\times$100, 150$\times$150 and 200$\times$200. Fig.d shows scaling in total time taken, with number of agents varied from 2 to 50 on a 100 $\times$100 grid. We also note that the $A^*$ algorithm is given a heuristic i.e. the Manhattan distance to the closest goal, information that the other algorithms are not given.}}
\label{fig:metrics}
\vspace{-10pt}
\end{figure}

\subsection{Results and Discussion}
\paragraph{\revision{Comparison with $A^*$:}} 
\revision{We include $A^*$ as a baseline heuristic search
algorithm. $A^*$ uses a heuristic based on Manhattan distance and
does not account for proximity to obstacles, while the other algorithms
incorporate this information into their reward structures. We note that each agent invokes $A^*$ separately to plan its path, while the other algorithms learn a common policy. The paths identified
by $A^*$ are close to optimal as the Manhattan distance
is an admissible heuristic, and thus a good baseline for comparison.}

\paragraph{\revision{Timesteps/Path length:}}
\revision{MAPF-EGT excels in path length (i.e. timesteps to reach the goal) for large grids. PPO and Monte Carlo rival it on smaller grids but lag behind on larger ones.}

\paragraph{\revision{Distance from obstacles:}} 
\revision{ Monte Carlo keeps a relatively larger distance from obstacles, showing a conservative approach. Q-Learning and PPO have similar distances, indicating a balanced approach in avoiding obstacles. MAPF-EGT shows a more adaptive strategy, maintaining sizeable distances from obstacles even as grid size increases.}

\paragraph{\revision{Computation time:}} 
\revision{Q learning struggles with
scalability and training, affecting its large-grid success rate. PPO takes longer time on smaller grids but scales well with increase in size. MAPF-EGT has a comparable computation time to the Monte-Carlo method.}

\paragraph{\revision{Number of agents:}} 
\revision{Testing MAPF-EGT with 2 to 100 agents on a 100x100 grid, we see that PPO scales rather poorly, while Q-learning remains steady. Although Q-learning had long training times as indicated previously, it's short run times make it scale well with increase in the number of agents. MAPF-EGT scales the best with increase in the number of agents. Our intuition for this is that more agents represent more trajectories generated, and more diverse and richer data to learn from.
}
% \revision{% \begin{enumerate}[label=(\alph*)]

%% file: related.tex
% \mypara{Related Work}
Early work in modeling environment dynamics used Markov decision processes (MDPs) \cite{SadighKapoor2016,HaesaertEtAl2018} or differential equations \cite{GilpinEtAl2020,PantEtAl2018}. Recently, the focus has shifted to data-driven and automata-based methods for control synthesis and achieving temporal logic-based objectives \cite{balakrishnan2023modelfree}. Modern techniques incorporate robustness metrics into deep reinforcement learning (RL) frameworks, replacing traditional reward functions with back-propagation to train controllers for complex temporal tasks \cite{balakrishnan2019structured}. In RL, reward engineering is crucial, with methods ranging from Q-learning for robust controller learning \cite{AksarayEtAl2016}, to using deterministic finite automata (DFA) for task specifications and potential functions. Automata are preferred for their robustness in symbolic weighted automata frameworks \cite{algebraicstl} and the ability to translate various specification formalisms into automata \cite{JothimuruganEtAl2021}. Addressing Non-Markovian Rewards (NMRs) now involves Non-Markovian Reward Decision Processes (NMRDPs) with history-sensitive reward functions \cite{CamachoEtAl2018}. Model-based MDP synthesis under Linear Temporal Logic (LTL) has been streamlined to mixed integer linear programming, using task progression concepts to develop policies \cite{KalagarlaJainNuzzo2020}. Advanced techniques for LTL controller synthesis handle unsatisfiable tasks \cite{GuoZavlanos2018}, and Temporal Logic Policy Search introduces a robustness-oriented approach to model-free RL \cite{LiMaBelta2018}.

%% file: conclusion.tex
\revision{In this paper, we have addressed the complex challenge of trajectory-level task objectives for autonomous multi-agent systems in uncertain environments. Our contributions include the introduction of weighted automata-based task objectives, which enhance the agents’ ability to complete tasks more expeditiously, and the development of the MAPF-EGT algorithm, which leverages evolutionary game theory to train homogeneous agent teams more effectively. Our empirical results demonstrate that our approach outperforms state-of-the-art reinforcement learning and heuristic search methods, achieving a reduction in path length and faster computation times. These findings highlight the potential of our methods to improve the efficiency and scalability of multi-agent systems in various practical applications.}

\paragraph{Limitations and future work:}
Our current setup is confined to discrete state-action spaces, but we are actively developing an extension to accommodate continuous spaces. This expansion will utilize function approximation through radial basis functions, enabling policy updates for states within proximity to the updated state. Furthermore, given our methodology of normalizing the probability distribution across a state's actions, we plan to implement a continuous model, which will be adjusted using probability density functions and updated via Dirac delta functions.
Our research currently focuses on homogeneous agents and tasks. Future developments will aim to include diverse forms of multi-agent learning, providing convergence guarantees. Additionally, we intend to explore game-theoretic guarantees within the realm of multi-agent learning to ensure robust and strategic interactions among agents.